\documentclass[seceq,preprint]{ptptex}

\usepackage{amsthm}

\newtheorem{prop}{Proposition}

\newcommand{\ket}[1]{\left|#1\right>}
\newcommand{\bra}[1]{\left<#1\right|}
\newcommand{\braket}[2]{\left<#1|#2\right>}

\newcommand{\zket}[1]{\left|#1\right>\!{\mathstrut}_0}
\newcommand{\zbra}[1]{{\mathstrut}_0\!\left<#1\right|}

\newcommand{\rmd}{{\rm d}}

\def \ap {\alpha^\prime}

\def \metric {\mathcal{G}}

\def \phys {{\rm phys}}

\def \ii {i}
\def \dd {\partial}

\def \QB  {Q_{\! B}} 

\def \tQB {\widetilde{Q}_{\! B}}

\newcommand{\tket}[1]{\!\left.\vphantom{#1}\right|\!\tilde{#1}\!\left>\vphantom{#1}\right.\!}

\def \nn {\nonumber \\}

\notypesetlogo                       
\preprintnumber[3cm]{
YITP-10-75
}

\title{
No-Ghost Theorem for\\ Neveu-Schwarz String in 0-Picture
}

\author{
Maiko \textsc{Kohriki},\footnote{%
E-mail:\  {\tt kohriki@yukawa.kyoto-u.ac.jp}} %
Hiroshi \textsc{Kunitomo}\footnote{%
E-mail:\  {\tt kunitomo@yukawa.kyoto-u.ac.jp}}
and %
Masaki \textsc{Murata}\footnote{%
E-mail:\  {\tt masaki@yukawa.kyoto-u.ac.jp}}%
}

\inst{
Yukawa Institute for Theoretical Physics, Kyoto University, \\
Kyoto 606-8502, Japan
}

\abst{
The no-ghost theorem for Neveu-Schwarz string is directly proved in 0-picture.
The one-to-one correspondence between physical states in $0$-picture and in the 
conventional $(-1)$-picture is confirmed. It is shown that a nontrivial metric
consistent with the BRST cohomology is needed to define 
a positive semidefinite norm in the physical Hilbert space.
As a by-product, we find a new inverse picture-changing operator,
which is noncovariant but has a nonsingular operator product with itself.
A possibility to construct a new gauge-invariant superstring field theory is discussed.}

\begin{document}

\maketitle


\section{Introduction}\label{sec:introduction}

String field theory\cite{Witten:1985cc,HIKKO,Neveu:1986mv} (SFT) is 
one of the promising candidates to provide a nonperturbative definition 
of string theory. For the bosonic cubic SFT\cite{Witten:1985cc}, in fact, 
an analytic solution describing tachyon condensation was found\cite{Schnabl:2005gv}
and used to prove\cite{Ellwood:2006ba} Sen's conjecture.\cite{Sen:1999mh}
Successively, many other analytic solutions were also found 
and used to study nonperturbative phenomena in string theory.\cite{Okawa:2006vm}\tocite{Hellerman:2008wp}

On the other hand, superstring field theory (SSFT) is not well understood so far.
The original cubic superstring field theory\cite{Witten:ssft} 
based on the Neveu-Schwarz-Ramond (NSR) formulation has the divergent contact term problem.\cite{Wendt:1987zh}
This problem is caused by colliding picture-changing operators inserted 
at the interaction point and breaks the gauge invariance. 
Then, in order to avoid this problem, a modified formulation based on the Neveu-Schwarz (NS) 
string field in 0-picture was proposed.\cite{Preitschopf:1989fc,Arefeva:1989cm}\footnote{
There is another candidate of SSFT, which is free from the contact term problem.\cite{Berkovits:1995ab} 
In this formulation, the picture-changing operator is absent but the gauge-invariant action is 
nonpolynomial. Some analytic solutions for marginal deformations in this SSFT were also 
found.\cite{Erler:2007rh}\tocite{Kiermaier:2007ki}}
Several of techniques of solving the classical equation of motion developed in the bosonic cubic SFT 
can also be applicable to this modified cubic SSFT. Actually, a number of analytic solutions have been 
obtained.\cite{Erler:2007xt}\tocite{Aref'eva:2010yd}
However, the gauge-fixing procedure of the modified cubic SSFT is not so clear owing to the kernel
of the picture-changing operator inserted in the kinetic term of the NS string field.
In addition, it has recently been pointed out that
there is another difficulty in the modified cubic SSFT.\cite{Kroyter:2009zi}
We eventually encounter the same divergence as the contact term problem
by repeating the gauge transformation with a Ramond gauge parameter. 
Hence, the finite gauge transformation cannot be defined.
This divergence is inevitable as long as the conventional picture-changing operator is used.

Here, we focus on the first difficulty. 
One of the reasons for this difficulty is that the physical states of the 0-picture NS string 
are not properly clarified yet. They are known only indirectly through the relation with 
those in the conventional ($-1$)-picture.
In this paper, we therefore study the structure of the physical Hilbert space of NS string 
and prove the no-ghost theorem directly in 0-picture. 
As its by-product, we find a new inverse picture-changing operator,
which is noncovariant but has a finite operator product with itself.
We propose a new SSFT using this new operator as a possible solution to 
the second difficulty. The gauge transformation in the new formulation
can be integrated to the finite gauge transformation without any obstruction.

The modern BRST quantization of the bosonic string theory was first performed 
by Kato and Ogawa.\cite{Kato:1982im} They proved that the physical Hilbert space
is isomorphic to the conventional one spanned by DDF states\cite{DDF:1971fp} 
and, hence, has a positive semidefinite norm. 
This norm of string states is, however, a little nontrivial owing to the ghost zero-modes.
While we concentrate on the NS string in this paper,
a similar no-ghost theorem was also proved for the NSR superstring.\cite{Ohta:1985af,Ito:1985qa}
In these papers, string states are constructed on the oscillator ground state
annihilated by negative frequency modes as usual. The no-ghost theorem can be proved 
in a completely parallel way to the case of the bosonic string.\cite{Kato:1982im} 
However, it is known that this ground state is not unique but there are infinitely 
many ground states labeled by an integer, so-called picture number.
They provide inequivalent representations of the canonical commutation relation for
the superconformal ghosts.\cite{Friedan:1985ge}
The conventional ground state is the one in ($-1$)-picture, called natural picture in this paper. 
The modified cubic SSFT is constructed using this degree of freedom.
The NS string field is defined using the Fock space in 0-picture. 
By construction, it is guaranteed that the correct on-shell physical amplitudes
are reproduced because they are independent of the picture chosen to be computed. 
However, the total Hilbert space in 0-picture, that is, the space of the off-shell string field,
seems to be quite different from the one in natural picture.
It is, therefore, worthwhile to prove the no-ghost theorem directly in 0-picture. 
We expect that it sheds light on the gauge-fixing problem of the modified cubic SSFT.

This paper is organized as follows. 
In \S\ref{sec:0pps}, the physical state conditions of the NS string are studied 
directly in 0-picture. We construct the physical states in a similar way to the bosonic 
string\cite{Kato:1982im} and show that there is a one-to-one correspondence between 0-picture 
physical states and conventional ones in natural picture. 
The no-ghost theorem is proved in \S\ref{sec:norm}. 
A natural norm of 0-picture states is induced from an inner product defined by introducing 
a metric consistent with the BRST cohomology. Then, we prove that the physical states have positive 
semidefinite norms. A discussion is given in \S\ref{sec:disc}.
We find a new inverse picture-changing operator and propose a new consistent SSFT. 
This new formulation with the new picture-changing operator is unfortunately noncovariant but 
still has gauge invariance. In addition, the new gauge transformation can be integrated to 
the finite transformation without any obstruction.
The oscillator conventions are summarized in Appendix \ref{app:one}.
Appendix \ref{app:two} is devoted to discussing the oscillator ground state and its inner product
in general picture.

\section{Physical states in 0-picture}\label{sec:0pps}

In this section, we investigate the NS string in 0-picture. After describing Fock states, 
the physical state condition is examined in detail. The physical states can be constructed 
by a similar method to that used for the bosonic string.\cite{Kato:1982im} 
We explicitly give the well known map between the physical states between 0- and ($-1$)-pictures.

\subsection{Physical state condition in 0-picture}\label{sec:0pps1}

Let us first separately consider the representations of
ghost zero-modes $(b_0,c_0)$ and the other oscillators. 
For the ghost zero-modes, 
we take a two-dimensional representation 
$(\ket{\downarrow},\ket{\uparrow})$ defined by\cite{Kato:1982im} 
\begin{subequations}
\begin{alignat}{2}
b_0 \ket{\downarrow}&=0,\qquad&  b_0 \ket{\uparrow}&=\ket{\downarrow}\,,\\
c_0 \ket{\downarrow}&=\ket{\uparrow},\qquad&  c_0 \ket{\uparrow}&=0\,.
\end{alignat}
\end{subequations}
The ground state of the NS string is, therefore, doubly degenerated.
For the other oscillators, the ground state $\ket{0}_0$ in 0-picture is defined by
\begin{align}\label{eq:0vac}
&\alpha_n^\mu \ket{0}_0 = b_n \ket{0}_0 = c_n \ket{0}_0 = 0 
\qquad {\rm for} \quad n > 0 \,, \nn
&\psi^\mu_r \ket{0}_0 =0 \qquad {\rm for} \quad r > 0 \,, \nn
&\beta_r \ket{0}_0 = 0 \qquad {\rm for}\quad r \ge -\tfrac12 \,,\nn
&\gamma_r \ket{0}_0 = 0 \qquad {\rm for}\quad r\ge \tfrac32 \,.
\end{align}
Combining with the momentum eigenstate, $p^\mu \ket{k} = k^\mu\ket{k}$,
it is useful to define 
\begin{equation}
 \zket{0,k}=\zket{0}\otimes\ket{k}.
\end{equation}
Here, it should be noted that canonical conjugate modes $\beta_{-\frac12}$ and $\gamma_{\frac12}$ 
exchange their role in this picture. 
The positive mode $\gamma_\frac12$ does not annihilate the ground state $\zket{0}$ but
the negative mode $\beta_{-\frac12}$ does.
The NS string state is in total written in the form of the direct product
\begin{equation}
\zket{\phi} = \tket{\phi} \otimes 
              \left\{ \ket{\downarrow}\text{or}\ket{\uparrow} \right\} \,,
\end{equation}
where $\tket{\phi}$ denotes a state in the Fock space $\mathcal{F}$ 
constructed on $\zket{0,k}$.

The physical state condition of the NS string\cite{Ohta:1985af,Ito:1985qa} is
then given by
\begin{subequations}
\begin{align}
 b_0\zket{\phys}&=0\,,\label{def:Siegel-gauge}\\
 \QB \zket{\phys}&=0\,,\label{def:physical_state}
\end{align}
\end{subequations}
with the BRST charge
\begin{subequations}
\begin{equation}
\QB = c_0 L + b_0 M + \sqrt{2\ap} \tQB \,,
\end{equation}
where
\begin{align}
L &= L_0^{(m)}
 +\sum_{n\ne0} n :b_{-n} c_n: + \sum_r r :\beta_{-r} \gamma_r:-1 \,,
\\
M &= -\sum_{n\ne0} n c_{-n}c_n - \sum_{r} \gamma_{-r}\gamma_r \,,
\\
\sqrt{2\ap} \tQB &= 
\sum_{n\neq0} c_{-n} L_n^{(m)} + \sum_r \gamma_{-r} G_r^{(m)} 
+ \sum_{\substack{n,m\neq0 \\ n+m\neq0}} \frac12(n-m) b_{-n-m}c_n c_m \nn
&\hspace{3cm}
+ \sum_{n\neq0} \sum_r \left( \frac12(2r-n) \beta_{-n-r}c_n
\gamma_r - b_{-n}\gamma_{n-r} \gamma_r \right) \,,
\end{align}
with
\begin{align}
L^{(m)}_n&=\frac12\sum_m:\alpha^\mu_{n-m}\alpha_{\mu m}:+\sum_r\frac14(2r-n):\psi^\mu_{n-r}\psi_{\mu r}:\,,\\
G^{(m)}_r&=\sum_n\psi^\mu_{r-n}\alpha_{\mu n}\,.
\end{align}
\end{subequations}
The normal ordering $:\ :$ is taken with respect to the 0-picture ground state (\ref{eq:0vac}).
The first condition (\ref{def:Siegel-gauge}), corresponding to the Siegel gauge condition in the context
of SSFT, imposes $\zket{\phys}=\tket{\phi}\otimes \ket{\downarrow}$. 
The second condition (\ref{def:physical_state}) on this state becomes
\begin{subequations}\label{def:phys_state_cond}
\begin{align}
L\tket{\phi} = 0 \,, \label{cond:mass-shell}\\
\tQB \tket{\phi}  =0 \,. \label{def:nonzero-phys_state}
\end{align}
\end{subequations}
In what follows, we first define the subspace $\mathcal{V}_L$ satisfying the on-shell condition (\ref{cond:mass-shell}).
Then, general physical states are constructed by solving the condition (\ref{def:nonzero-phys_state}) 
within this subspace.

\subsection{Construction of physical states}   

In the beginning, the ground state $\zket{0,k}$ satisfies (\ref{cond:mass-shell})
if the tachyonic on-shell condition $\ap k^\mu k_\mu=1$ holds.
We introduce the light-cone coordinates $u_\pm = \frac{1}{\sqrt2} (\pm u_0 + u_9)$ 
and choose a Lorentz frame in which the transverse momentum $k_i\ (i=1,\cdots,8)$ 
is equal to zero. In this frame, 
the on-shell momentum for the ground state can be written as 
\begin{equation}
 k^\mu=k^\mu_0\equiv
\left(k^+\,,\  k^-=\frac{1}{2\ap k^+}\,,\ k_i=0\right)\,,
\end{equation}
where $k^+\ne0$ is assumed. The subspace $\mathcal{V}_L$ is constructed on this on-shell
ground state $\zket{0,k_0}$
as the Fock space spanned by
the rescaled oscillators
\begin{equation}\label{def:rescaling-oscillators}
 \hat{\varphi}_{-q}=e^{-i\frac{q}{2\ap p^+}x^+}\varphi_{-q}\,,
\end{equation}
where $\varphi_{-q}$ denotes oscillators $\alpha^\mu_{-n},\ b_{-n},\ c_{-n},\ 
\psi^\mu_{-r},\ \beta_{-r}$ or $\gamma_{-r}$.
These new oscillators commute with $L$ and, hence, the states created 
by them on $\zket{0,k_0}$ automatically satisfy (\ref{cond:mass-shell}). 

Let us next examine the second condition (\ref{def:nonzero-phys_state}).
Following the conventional method\cite{Kato:1982im,Ohta:1985af,Ito:1985qa},
we take another rescaling
\begin{equation}
\ap \to \frac{\ap}{\hbar^2} \,, \qquad 
p^- \to \hbar^2 p^- \,, \qquad 
x^+ \to \frac{x^+}{\hbar^2}\,,\label{def:hbarrescaling}
\end{equation}
introducing a formal expansion parameter $\hbar$. 
As a consequence of this rescaling, $\tQB$ is separated into three pieces by the order of $\hbar$ 
as 
\begin{subequations}
\begin{align}
\tQB(\hbar) &= A + \hbar B + \hbar^2 C \,, \label{def:rescaledqbtilde}\\
A &= \sum_{n\ne0} p^+c_{-n}\alpha^-_n 
+\sum_rp^+\gamma_{-r}\psi^-_r \,,
\label{def:tQBA}\\
B &= 
\frac{1}{2\sqrt{2\ap}} \Bigg(
\sum_{\substack{n,m\neq0 \\ m\neq n}}  
 \Big( c_{-n}\alpha^\mu_{n-m}\alpha_{\mu m} + (n-m) b_{n-m}c_nc_m \Big) \nonumber\\
&\hspace{15mm}
+ \sum_{\substack{n\neq0\\r}}
 \Big( \frac12(2r-n)  c_{-n}\psi_{n-r}^\mu\psi_{\mu r} + (2r+n)  c_{-n}\beta_{n-r} \gamma_r  \nonumber\\
&\hspace{55mm}
+ 2 \gamma_r\alpha_{-n}^\mu\psi_{\mu n-r} -2 b_{-n}\gamma_{n-r}\gamma_r
\Big)
\Bigg) \,,
\\
C &= \sum_{n\ne0} p^-c_{-n}\alpha^+_n
+\sum_rp^-\gamma_{-r}\psi^+_r \,.
\end{align}
\end{subequations}
Here, the symbol\hspace{2mm}  $\hat{}$\hspace{2mm} on the oscillators
is omitted because $\widetilde{Q}_B$ is invariant under the oscillator rescaling
(\ref{def:rescaling-oscillators}). 
If we assume that the state $\tilde{\ket{\phi}}$ can also be expanded as
$|\tilde{\phi}(\hbar)\rangle = \sum_{n=0}^\infty \hbar^n \tilde{\ket{\phi}}^{(n)}$ and
the condition
\begin{equation}
 \tQB(\hbar)|\tilde{\phi}(\hbar)\rangle=0
\end{equation}
holds order by order in $\hbar$, equations
\begin{equation}
A \tilde{\ket{\phi}}^{(n)} + B\tilde{\ket{\phi}}^{(n-1)} + C\tilde{\ket{\phi}}^{(n-2)} = 0 \,,
\qquad \textrm{for}\qquad n\ge0\,,\label{eq:cohomology}
\end{equation}
are obtained with the understanding that $\tilde{\ket{\phi}}^{(n)}\equiv0$ for $n<0$. 
The physical states can be obtained by solving Eq. (\ref{eq:cohomology}) iteratively.

Let us begin with finding the physical tachyon state.
In the lowest order, Eq. (\ref{eq:cohomology}) is simply
\begin{equation}
 A\tilde{\ket{\phi}}^{(0)}=0\,.\label{eq:cohomology0}
\end{equation}  
The on-shell ground state $\zket{0,k_0}$ does not satisfy this condition   
because it is not annihilated by $\gamma_{\frac12}$ in 0-picture:
\begin{equation}
A \zket{0,k_0} = k^+ \gamma_{\frac12} \psi^-_{-\frac12} \zket{0,k_0} \neq 0 \,.
\end{equation}
The physical tachyon state is in the first excited level and given in this order by
\begin{equation} 
\zket{{\rm tach}}^{(0)}=\hat{\psi}^-_{-\frac12}\zket{0,k_0} \, .\label{def:tachyon^0}
\end{equation}  
Equation (\ref{eq:cohomology}) can then be solved in turn.  
The physical tachyon state, with rescaling (\ref{def:hbarrescaling}), is 
finally obtained as
\begin{align} \label{def:tach-hbar}
\zket{\rm tach(\hbar)}&=\left(
\hat{\psi}^-_{-\frac12} 
- \frac{\hbar}{\sqrt{2\ap}k^+}\hat{b}_{-1}\hat{\gamma}_{\frac12} 
+ \frac{\hbar^2}{4\ap (k^+)^2} \hat{\psi}^+_{-\frac12}
\right) \zket{0,k_0} \,,
\nn
&= \left(
\psi^-_{-\frac12} 
- \frac{\hbar}{\sqrt{2\ap}k^+}b_{-1}\gamma_{\frac12} 
+ \frac{\hbar^2}{4\ap (k^+)^2} \psi^+_{-\frac12}
\right) \zket{0,k_1} \,,
\end{align} 
where the momentum $k_1$ is given by
\begin{equation}
 k_1^\mu=\left(k^+,\frac{1}{4\ap k^+},0,\cdots,0\right)\,,
\end{equation}
and satisfies the correct on-shell condition for the physical tachyon, $\ap k_1^\mu k_{1\mu}=1/2$.

%
%
In the lowest order, we can easily see that general physical states can be constructed on 
$\zket{\rm tach}^{(0)}$ as
\begin{equation}
a^{i_1\dag}_{n_1}a^{i_2\dag}_{n_2}\dotsm \boldsymbol{\psi}^{j_1\dag}_{r_1}
\boldsymbol{\psi}^{j_2\dag}_{r_2}\dotsm \zket{\rm tach}^{(0)}\,,\label{def:physicalstates}
\end{equation}
because the transverse oscillators commute with $A$.
We introduce here the creation operators
\begin{subequations}
\begin{alignat}{2}
 a_n^{i\dag}&=\frac{\hat{\alpha}_{-n}^i}{\sqrt{n}}\,,&\qquad &(n>0)\\
 \boldsymbol{\psi}_r^{i\dag}&=\hat{\psi}^i_{-r}\,,&\qquad &(r>0)
\end{alignat}
\end{subequations}
for convenience.
As in the case of the tachyon, we can iteratively solve Eq. (\ref{eq:cohomology})
in principle. 
In fact, however, it is unnecessary because the problem is the same as
the one in natural picture if we take the complete tachyon state $\zket{{\rm tach}(\hbar)}$
as the oscillator \textit{ground} state\footnote{
We will see in \S\ref{sec:norm} that the tachyon state should be reinterpreted as the ground 
state also from the study of the quartet structure.}
without expanding in $\hbar$.
The solutions can be immediately written as
\begin{equation} \label{def:0-pic_phys_state}
A^{i_1\dag}_{n_1}(\hbar)A^{i_2\dag}_{n_2}(\hbar)\dotsm B^{j_1\dag}_{r_1}(\hbar)
B^{j_2\dag}_{r_2}(\hbar)\dotsm \zket{\rm tach(\hbar)}\,,
\end{equation}
with the well known DDF operators\cite{Brower:1973iz,Goddard:1972iy,Schwarz:1974ix} 
\begin{subequations}\label{def:DDF}
\begin{align}
A^{i\dag}_n (\hbar)
&=\frac{1}{\sqrt{n}}\oint\frac{\rmd z}{2\pi\ii} z^{-n} \left(i\partial X^i-\frac{n\hbar}{\sqrt{2\ap}p^+}\psi^+\psi^i\right)
e^{-i\frac{n\hbar}{\sqrt{2\ap}p^+}X^+(z)}\,,\quad (n>0)
\\
B^{i\dag}_r(\hbar)
&=\frac{1}{\sqrt{\alpha_0^+}} \oint\frac{\rmd z}{2\pi\ii} z^{-r} \left(
\psi^i(i\partial X^+)^{\frac12}-\frac{\psi^+i\partial X^i}{(i\partial X^+)^{\frac12}}
+\frac12 \frac{\psi^i\psi^+\partial\psi^+}{(i\partial X^+)^{\frac32}}
\right)e^{-i\frac{r\hbar}{\sqrt{2\ap}p^+}X^+(z)}\,,\nonumber\\
&\hspace{10.2cm} (r>0)
\end{align}
\end{subequations}
rescaled by (\ref{def:hbarrescaling}).\footnote{
Here, the fractional power of $i\partial X^+$ is defined by expansion around the zero mode
$\alpha^+_0=\sqrt{2\ap}p^+/\hbar$ by assuming $p^+\ne0$.}
One can easily see that the states (\ref{def:0-pic_phys_state}) are really desired physical states
by noting that the DDF operators (\ref{def:DDF}) are (anti-) commutative with $\tQB(\hbar)$
\begin{equation}
 [\tilde{Q}_B(\hbar),A^{\dag i}_n(\hbar)]=\{\tilde{Q}_B(\hbar),B^{\dag i}_r(\hbar)\}=0,
\end{equation}
and satisfy the initial condition 
$A^{i\dag}_n(0)=a^{i\dag}_n$ and $B^{i\dag}_r(0)=\boldsymbol{\psi}^{i\dag}_r$.
It is also noted here that the DDF operators satisfy canonical 
(anti-)commutation relations,
\begin{equation}\label{def:ccr}
\left[ A_n^i(\hbar)\,, A_m^{j\dagger}(\hbar) \right] = 
\left\{ B_r^i(\hbar)\,, B_s^{j\dagger}(\hbar) \right\} = \delta_{i,j}\delta_{n,m} \,.
\end{equation}

Because the physical states (\ref{def:0-pic_phys_state}) appear similar to
the conventional ones in natural picture, one may consider the relation between them.
As a matter of fact, 
they must be identical because the physical spectrum is independent of the picture.
We can answer this question with the help of the well known
picture-changing operator\footnote{
A precise definition of this delta function is discussed in Appendix \ref{app:two}.}
\begin{align}\label{def:pco_X}
X(z) &= \{\QB\,, \Theta(\beta(z))\} \nonumber \\
 &=G(z)\delta(\beta(z))-\partial b\delta'(\beta(z)),
\end{align}
where
\begin{equation}
 G(z)=G^{(m)}(z)+c(z)\partial\beta(z)+\frac32\partial c(z)\beta(z)-2\gamma(z)b(z).
\end{equation}
Using $X(z)$, rescaled by (\ref{def:hbarrescaling}),
one can show that
\begin{align}
\sqrt{2\ap} k^+ \zket{\rm tach(\hbar)} 
&= \left( 
\sqrt{2\ap} p^+ \hat{\psi}^-_{-\frac12} 
- \hbar \hat{b}_{-1} \hat{\gamma}_{\frac12} 
+ \hbar^2 \sqrt{2\ap} p^- \hat{\psi}^+_{-\frac12} 
\right)\delta(\hat{\beta}_{-\frac12}) \ket{0,k_1}_{-1} \nonumber\\
&= \hbar X(0) \ket{0,k_1}_{-1} \,,
\end{align}
where $\ket{0,k_1}_{-1}=\delta(\gamma_{\frac12})\zket{0,k_0}$
is the physical tachyon in $(-1)$-picture. 
Considering the fact that $X(z)$ commutes with the DDF operators,
this gives the well known map between the physical states in 0- and $(-1)$-pictures.\cite{Friedan:1985ge}
In order to prove the no-ghost theorem, however, we must also demonstrate that 
all the physical states have the form of (\ref{def:0-pic_phys_state}),
which is far from trivial.

\section{No-ghost theorem}\label{sec:norm}

\subsection{Definition of a norm for $0$-picture states}

Let us first define an inner product in a manner consistent with the physical state 
condition (\ref{def:phys_state_cond}).
As explained in Appendix \ref{app:two} in detail, a nontrivial metric $\metric$
is needed to define a consistent inner product with the BRST cohomology.
The inner product between two 0-picture states 
$\zket{\alpha}$ and $\zket{\beta}$ is defined using the nontrivial metric $\metric$ by 
$\zbra{\alpha}\metric\zket{\beta}$, which naturally induces a norm of $\zket{\alpha}$ as
\begin{equation}
  \Big\|\zket{\alpha}\Big\|=\zbra{\alpha}\metric\zket{\alpha}\,.
\end{equation}
The first task is therefore to find such a metric
that satisfies the conditions 
\begin{subequations}\label{cond:metric}
\begin{align}
[L\,, \metric] &= 0 \,, \label{cond:metric-L}\\
[\tQB\,, \metric] &= 0 \,,  \label{cond:metric-tQB}\\ 
[A^{i\dag}_n,\metric]=0\,,\quad&\quad [B^{i\dag}_r,\metric]=0\,, \label{cond:metric-DDF}\\
\zbra{\rm tach} \metric \zket{\rm tach}&\gneqq0 \,. \label{cond:metric-norm}
\end{align}
\end{subequations}
The first two, (\ref{cond:metric-L}) and (\ref{cond:metric-tQB}), are
required for the consistency with the physical state condition (\ref{def:phys_state_cond}).
In particular, the condition \eqref{cond:metric-tQB} guarantees that $\tQB$-exact states
have zero norm and are orthogonal to any physical state. 
The last two, \eqref{cond:metric-DDF} and \eqref{cond:metric-norm}, ensure that the physical states
\eqref{def:0-pic_phys_state} have positive semidefinite norms.

Similar to the analysis of the physical states in the previous section, we assume that
the metric $\metric$ can also be expanded in $\hbar$ as
$\metric(\hbar) = \sum_{n=0}^\infty \hbar^n \metric^{(n)}$\,.
Consequently, the conditions (\ref{cond:metric}) are also expanded in $\hbar$ and must 
hold in each order of $\hbar$. 
In particular, from the second condition \eqref{cond:metric-tQB}, we obtain
\begin{equation} 
[A\,, \metric^{(n)}]+[B\,, \metric^{(n-1)}]+[C\,, \metric^{(n-2)}]=0,\qquad
{\rm for}\qquad n\ge0\,,
\label{cond:metric-tQB-hbar-n}
\end{equation}
with the understanding that $\metric^{(n)}\equiv0$ for $n<0$.
This equation becomes simply 
\begin{equation}
[A\,,\metric^{(0)}]=0\,,  
\end{equation}
for $n=0$ and is satisfied by
\begin{equation}
 \metric^{(0)}=\psi^+_{-\frac12}\delta(\gamma_{-\frac12})
\delta(\gamma_\frac12)\psi^+_\frac12\,.\label{def:metric0}
\end{equation}
This is the most natural choice also satisfying (\ref{cond:metric-L}) and (\ref{cond:metric-norm})
at the lowest order. Note here that this metric has a nontrivial kernel and therefore 
is degenerate. This degeneracy does not produce any difficulties but is necessary for consistency
as explained later.  
Starting with this lowest-order solution (\ref{def:metric0}),
we can solve Eq. (\ref{cond:metric-tQB-hbar-n}) one after another.
The result is eventually written in a closed form as
\begin{subequations}\label{def:metric_closed_form}
\begin{equation} \label{def:metric}
\metric(\hbar) = \Psi^+_{-\frac12}(\hbar) \, \delta\!\left(\Gamma_{-\frac12}(\hbar)\right) 
\delta\!\left(\Gamma_{\frac12}(\hbar)\right) \, \Psi^+_{\frac12}(\hbar) \,,
\end{equation}
with
\begin{align}
\Psi^+_{\frac12}(\hbar) &= \left(\Psi^+_{-\frac12}(\hbar)\right)^\dag\nonumber\\
&=\frac{1}{\left(\alpha_0^+\right)^\frac12} \oint \frac{\rmd z}{2\pi\ii} 
z^{\frac12}\psi^+(z) \left( \ii\dd X^+(z) \right)^{\frac12} \,, \\
\Gamma_{\frac12}(\hbar) 
&= \left(\Gamma_{-\frac12}(\hbar)\right)^\dag\nn
&=\frac{1}{\left(\alpha_0^+\right)^\frac32} \oint \frac{\rmd z}{2\pi\ii} 
z^{\frac12}\bigg( \gamma(z) \left( \ii\dd X^+(z) \right)^{\frac32} 
 -\frac{1}{2} \gamma(z)\psi^+\dd\psi^+(z) \left( \ii\dd X^+(z) \right)^{-\frac12} \nonumber \\
&\hspace{60mm}
 -\frac{1}{2z}c(z) \psi^+(z) \left( \ii\dd X^+(z) \right)^{\frac12}\bigg) \,,
\end{align}
\end{subequations}
where $X^+(z)$ is rescaled using (\ref{def:hbarrescaling}).
We note that the term including $c_0$ in $\Gamma_\frac12(\hbar)$ is proportional to
$\Psi_\frac12^+(\hbar)$ and, hence, does not appear in the metric (\ref{def:metric}).

\subsection{No-ghost theorem for $0$-picture NS string}\label{subsec:no-ghost}

Now, we are ready to prove the no-ghost theorem for the 0-picture NS string.
The following proposition is first proved as the basis of mathematical induction.
\begin{prop} 
Any state $\zket{\phi}\in{\cal V}_L$ satisfying $A\zket{\phi}=0$ can be written as
\begin{equation}
\zket{\phi} = P^{(0)} \zket{\phi} + A \zket{\rho} \,,
\end{equation}
where $P^{(0)}$ is the projection operator onto the subspace generated by the transverse modes
\begin{equation}
P^{(0)} 
=\sum_{\{k\},\{l\}}\prod_{n,i}\frac{1}{k_n^i!}(a_n^{i\dag})^{k_n^i}
\prod_{m,j}^{\longrightarrow}(\boldsymbol{\psi}_r^{j\dag})^{l_m^j}\zket{\rm tach}^{(0)}
{\mathstrut}^{(0)}\!\!\!\zbra{\rm tach} \metric^{(0)} 
\prod_{m,j}^{\longleftarrow}(\boldsymbol{\psi}_r^{j})^{l_m^j}
\prod_{n,i}\frac{1}{k_n^i!}(a_n^{i})^{k_n^i}\,,
\end{equation}
where the arrow above the symbol $\prod$ indicates the ordering of product.
\end{prop}

The proof is almost the same as the standard one\cite{Kato:1982im,Ohta:1985af,Ito:1985qa}
but some important modification is needed. 
We, therefore, give a proof in some detail with careful attention to the differences.
\begin{proof}
As already mentioned, $\hat{\beta}_{-\frac12}$ and $\hat{\gamma}_{\frac12}$ exchange their role in 0-picture. 
In order to make it consistent with the quartet structure shown below,
the roles of $\hat{\psi}^-_{-\frac12}$ and $\hat{\psi}^+_\frac12$ must also be exchanged.
As well known in the hole theory, such an exchange occurs for the fermionic oscillator  
if the level is occupied. For the case considered here, it is accomplished by interpreting 
the tachyon state
\begin{equation}
\zket{\rm tach}^{(0)}=\hat{\psi}^-_{-\frac12}\zket{0,k_0}\,,  
\end{equation}
as the Fock vacuum  annihilated by $\hat{\psi}^-_{-\frac12}$.
The operator $\hat{\psi}^+_\frac12$ creates a hole on this vacuum.\footnote{
This hole has negative energy but causes no difficulty
because it is decoupled from the physical states as a quartet.}
In view of this, we define annihilation operators
\begin{subequations}\label{def:annihilations}
\begin{alignat}{3}
 \varphi_n&=-\frac{\ii}{\sqrt{n}p^+} \hat{\alpha}_n^+ \,, &\qquad
\eta_n&=\sqrt{n} \hat{c}_n \,, &\qquad (n&>0) \\
\bar{\eta}_n&=\frac{1}{\sqrt{n}} \hat{b}_n \,,&\qquad 
\tilde{\varphi}_n&=\frac{p^+}{\sqrt{n}} \hat{\alpha}_n^- \,,&\qquad (n&>0) \\
\chi_r&=\frac{1}{p^+} \hat{\psi}_r^+ \,,&\qquad 
\omega_r&=\hat{\gamma}_r \,,&\qquad (r&\ge\tfrac32)  \\
\bar{\omega}_r&=i \hat{\beta}_r  \,,&\qquad
\tilde{\chi}_r&=p^+ \hat{\psi}_r^- \,,&\qquad (r&\ge-\tfrac12)  
\end{alignat}
\end{subequations}
and creation operators
\begin{subequations}\label{def:creations}
\begin{alignat}{3}
 \varphi_n^\dag&=\frac{\ii}{\sqrt{n}p^+} \hat{\alpha}_{-n}^+ \,, &\qquad
\eta_n^\dag&=\sqrt{n} \hat{c}_{-n} \,, &\qquad (n&>0) \\
\bar{\eta}_n^\dag&=\frac{1}{\sqrt{n}} \hat{b}_{-n} \,,&\qquad 
\tilde{\varphi}_n^\dag&=\frac{p^+}{\sqrt{n}} \hat{\alpha}_{-n}^- \,,&\qquad (n&>0) \\
\chi_r^\dag&=\frac{1}{p^+} \hat{\psi}_{-r}^+ \,,&\qquad 
\omega_r^\dag&=\hat{\gamma}_{-r} \,,&\qquad (r&\ge-\tfrac12)  \\
\bar{\omega}_r^\dag&=i \hat{\beta}_{-r}  \,,&\qquad
\tilde{\chi}_r^\dag&=p^+ \hat{\psi}_{-r}^- \,.&\qquad (r&\ge\tfrac32)
\end{alignat}
\end{subequations}
Their (anti-) commutation relations with $A$ and themselves can be computed as
\begin{subequations}\label{def:quartets}
\begin{alignat}{2}
[A,\varphi_n]&=i\eta_n\,,&\qquad
\{A,\eta_n\}&=0\,,\\
\{A,\bar{\eta}_n\}&=\tilde{\varphi}_n\,,&\qquad
[A,\tilde{\varphi}_n]&=0\,,\\
\{A,\chi_r\}&=\omega_r\,,&\qquad
[A,\omega_r]&=0\,,\\
[A,\bar{\omega}_r]&=i\tilde{\chi}_r\,,&\qquad
\{A,\tilde{\chi}_r\}&=0\,,
\end{alignat}
and
\begin{equation}
 [\phi_i,\phi_j^\dag \}=
\genfrac{}{}{0mm}{}
{
\begin{matrix}
\hspace{5mm}\varphi_m^\dag\hspace{1.5mm}&\hspace{2mm}\tilde{\varphi}_m^\dag\hspace{1.5mm}&\hspace{2mm}
 \eta_m^\dag\hspace{1.5mm}&\hspace{1mm}\bar{\eta}_m^\dag\hspace{1mm}&\hspace{1mm}
 \chi_s^\dag\hspace{1mm}&\hspace{1mm}\tilde{\chi}_s^\dag\hspace{1.5mm}&\hspace{1.5mm}
 \omega_s^\dag\hspace{1.5mm}&\hspace{1mm}\bar{\omega}_s^\dag
\end{matrix}
}
{
\begin{matrix}
 \varphi_n \\
 \tilde{\varphi}_n \\
 \eta_n \\
 \bar{\eta}_n \\
 \chi_r \\
\tilde{\chi}_r \\
 \omega_r \\
 \bar{\omega}_r
\end{matrix}
\begin{pmatrix}
0 &-i\delta_{n,m} & & & & & & \\
i\delta_{n,m} &0 & & & & & & \\
 & &0 &\delta_{n,m} & & & & \\
 & &\delta_{n,m} &0 & & & & \\
 & & & &0 &\delta_{r,s} & & \\
 & & & &\delta_{r,s} &0 & & \\
 & & & & & &0 &i\delta_{r,s} \\
 & & & & & &-i\delta_{r,s} &0 
\end{pmatrix},
}
\end{equation}
\end{subequations}
respectively. 
These relations (\ref{def:quartets}) show that the nontransverse modes form two quartets
$(\varphi_n, \tilde{\varphi}_n, \eta_n, \bar{\eta}_n)$ and
$(\chi_n, \tilde{\chi}_n, \omega_n, \bar{\omega}_n)$.
Here, we must note that there is an asymmetry in the latter quartet. 
That is, annihilation operators $(\chi_{\pm\frac12}, \omega_{\pm\frac12})$ and
creation operators $(\bar{\omega}^\dag_{\pm\frac12}, \tilde{\chi}^\dag_{\pm\frac12})$
are missing. This is apparently inconsistent with the hermite conjugate relation 
(\ref{def:hermite_conjugate}). The states created by $(\chi^\dag_{\pm\frac12}, \omega^\dag_{\pm\frac12})$, 
however, are in the kernel of the metric $\metric^{(0)}$ and, hence, orthogonal to all the states 
in the Fock space. The remaining operators with $r\ge3/2$ are proper quartets including 
the inner product structure. The projection operator $P^{(n)}$ onto 
the $n$ quartet mode can be defined as
\begin{subequations}
\begin{align}
P^{(n)}=&
\frac1n\Bigg(\sum_k\left(i\tilde{\varphi}_k^\dag P^{(n-1)} \varphi_k
       -i\varphi_k^\dag P^{(n-1)} \tilde{\varphi}_k
       +\bar{\eta}_k^\dag P^{(n-1)} \eta_k
       +\eta_k^\dag P^{(n-1)} \bar{\eta}_k\right)\nn
&\hspace{12mm}
       +\sum_s\Big(\tilde{\chi}_s^\dag P^{(n-1)} \chi_s
       +\chi_s^\dag P^{(n-1)} \tilde{\chi}_s
       -i\bar{\omega}_s^\dag P^{(n-1)} \omega_s
       +i\omega_s^\dag P^{(n-1)} \bar{\omega}_s\Big)\Bigg)\,,
\nn
=&\{A, R^{(n)}\}\,,\qquad (n\ge1) \\
R^{(n)}=&\frac1n\Bigg(\sum_k\left(i\bar{\eta}_k^\dag P^{(n-1)} \varphi_k
   -i\varphi_k^\dag P^{(n-1)} \bar{\eta}_k\right)
  -\sum_s\Big(i\bar{\omega}_s^\dag P^{(n-1)}\chi_s
  -i\chi_s^\dag P^{(n-1)}\bar{\omega}_s\Big)\Bigg).
\end{align}
\end{subequations}
Using these projection operators, we can deduce that the general solution
to $A\zket{\phi}=0$ can be written as
\begin{align}
 \zket{\phi}&=\sum_{n=0}^\infty P^{(n)}\zket{\phi}\,,\nonumber\\
&=P^{(0)}\zket{\phi}+\sum_{n=1}^\infty\{A,R^{(n)}\}\zket{\phi}\,,\nonumber\\
&=P^{(0)}\zket{\phi}+A\zket{\rho}\,,
\end{align}
where $\zket{\rho}=\sum_{n=1}^\infty R^{(n)}\zket{\phi}$.
\end{proof}
We can prove the following second proposition by extending this result 
to all orders in $\hbar$.
\begin{prop} \label{prop:2}
Any state $\zket{\phi(\hbar)}\in{\cal V}_L$ satisfying $\tQB(\hbar)\zket{\phi(\hbar)}=0$ can be written as
\begin{equation} \label{eq:physstate-proj}
\zket{\phi(\hbar)} = P(\hbar) \zket{\phi(\hbar)} + \tQB(\hbar) \zket{\rho(\hbar)} 
\quad \exists \zket{\rho(\hbar)}\in {\cal V}_L \,,
\end{equation} 
where $P(\hbar)$ is the projection operator onto the subspace generated 
by the DDF operators $A_n^i(\hbar)$ and $B_r^i(\hbar)$,
\begin{align}
P(\hbar) 
=&\sum_{\{k\},\{l\}}\prod_{n,i}\frac{1}{k_n^i!}(A_n^{i\dag}(\hbar))^{k_n^i}
\prod_{m,j}^{\longrightarrow}(B_r^{j\dag}(\hbar))^{l_m^j}\zket{{\rm tach}(\hbar)}\nonumber\\
&\hspace{3cm}\times
\zbra{{\rm tach}(\hbar)} \metric(\hbar)
\prod_{m,j}^{\longleftarrow}(B_r^{j}(\hbar))^{l_m^j}
\prod_{n,i}\frac{1}{k_n^i!}(A_n^{i}(\hbar))^{k_n^i}\,.
\end{align}
\end{prop}
\begin{proof}
The proof given in the case of the bosonic string\cite{Kato:1982im} still remains valid,
because all the equations obtained by expanding in $\hbar$ are identical.
\end{proof}

Using the Proposition \ref{prop:2}, the no-ghost theorem can be easily proved.
That is, putting $\hbar=1$ in (\ref{eq:physstate-proj}), we can see that
any physical state has a positive semidefinite norm:
\begin{align}
\big\| \zket{\phi} \big\| &= 
\zbra{\phi} \metric \zket{\phi} \nonumber\\
&=\zbra{\phi}P^\dag\metric P\zket{\phi}\,,\nonumber\\
&=\zbra{\phi}P^\dag\metric\zket{\phi}= \zbra{\phi} \metric P \zket{\phi} \geq 0 \,,
\end{align}
using the (anti-)commutation relation (\ref{def:ccr}) and the fact that
the metric $\metric(\hbar)$ commutes with DDF operators (\ref{def:DDF}).
Hence, we conclude that the physical subspace of the 0-picture NS string
defined by (\ref{def:phys_state_cond}) is free from ghost states.

\section{Discussion}\label{sec:disc} 
%
%

In this paper, we constructed a nontrivial metric to define a consistent 
inner product with the BRST cohomology in 0-picture.
Inspired by this metric, we find a new BRST invariant inverse picture-changing operator
\begin{subequations}
\begin{equation}
Y_{\rm nc}(z) = \frac{\psi^+(z)}{i\partial X^+(z)}\delta\left(\Gamma\right(z))\,,
\end{equation}
with
\begin{equation}
 \Gamma(z)= \gamma(z) + c\frac{\partial\psi^+}{i\partial X^+}(z)\,.
\end{equation}
\end{subequations}
This $Y_{\rm nc}(z)$ is a conformal primary operator and an inverse of $X(z)$ in the sense of 
\begin{equation}
 \lim_{z\rightarrow w}Y_{nc}(z)X(w)=1\,.
\end{equation}
This new picture-changing operator is unfortunately Lorentz noncovariant but
has an important advantage. That is, it has a finite 
operator product with itself:
\begin{equation}
\lim_{z\to w}Y_{\rm nc}(z)Y_{\rm nc}(w) 
=\frac{\partial\psi^+\psi^+}{(i\partial X^+)^2}\delta(\partial\Gamma)\delta(\Gamma)(w)
\equiv (Y_{\rm nc})^2(w)\,.\label{eq:non-sing_ope}
\end{equation}
We can similarly obtain
\begin{align}
 (Y_{\rm nc})^n(z) &= \frac{\partial^{n-1}\psi^+\cdots\partial\psi^+\psi^+}{(i\partial X^+)^n}
\delta(\partial^{n-1}\Gamma)\cdots\delta(\partial\Gamma)\delta(\Gamma)(z)\,.
\end{align}
This yields a possibility to construct a new SSFT described by the action
\begin{equation}
S= \frac12 \int Y_{\rm nc}^2 A * \QB A 
+ \frac13\int Y_{\rm nc}^2 A * A * A 
+ \frac12 \int Y_{\rm nc}^3 \Psi * \QB \Psi 
+ \int Y_{\rm nc}^3 A * \Psi * \Psi \,,
\end{equation}
where  $A$ and $\Psi$ are 0-picture NS and $\frac{1}{2}$-picture Ramond 
string fields, respectively.
The operator $Y_{\rm nc}^n$ is inserted at the midpoint.  
This action is invariant under the gauge transformation
\begin{align}
\begin{split}
 \delta A=& Q_B\lambda+A*\lambda-\lambda*A+Y_{\rm nc}(\Psi*\epsilon-\epsilon*\Psi),\\
 \delta\Psi=& Q_B\epsilon+\Psi*\lambda-\lambda*\Psi+A*\epsilon-\epsilon*A,
\end{split}
\end{align}
where $\lambda$ and $\epsilon$ are 0-picture NS and $\frac12$-picture Ramond gauge 
transformation parameters, respectively. 
Owing to the nonsingular operator product (\ref{eq:non-sing_ope}),
this infinitesimal transformation can be integrated to the finite gauge transformation.

While this new SSFT is not manifestly Lorentz covariant,
the noncovariance is expected to disappear in the physical quantities, 
which should be independent of the choice of the picture-changing operator.
We hope that this new formulation gives a solution to 
the second difficulty\cite{Kroyter:2009zi} mentioned in \S\ref{sec:introduction}. 
This new SSFT is now under investigation.\cite{KKKM}

\section*{Acknowledgements}
The authors would like to thank Taichiro Kugo for useful discussions.
This work was supported in part by the Grant-in-Aid for the Global COE
program ``The Next Generation of Physics, Spun from Universality and Emergence''
from the Ministry of Education, Culture, Sports, Science and Technology (MEXT) of Japan.
The work of HK was supported in part by the Grant-in-Aid for Scientific Research
No. 19540284 from MEXT of Japan. The works of MK (No. 21-2291) and MM (No. 21-173) were
supported by Grants-in-Aid for Japan Society for the Promotion of Science (JSPS) Fellows.

\appendix
\section{Oscillator Conventions}\label{app:one} 

In this paper, we use oscillator conventions
\begin{subequations}
\begin{alignat}{2}
X^\mu(z,\bar{z})&=X^\mu(z)+X^\mu(\bar{z}),\nonumber\\
X^\mu(z)&=\frac{1}{\sqrt{2\ap}}x^\mu-i\sqrt{2\ap}&p^\mu\log z
&+i\sum_{n\ne0}\frac{\alpha^\mu_n}{n}z^{-n}, \\
\psi^\mu(z)&=\sum_{r} \psi^\mu_r z^{-r-\frac12},& &\\
b(z)&=\sum_n b_n z^{-n-2}, &
c(z)&=\sum_n c_n z^{-n+1}, \\
\beta(z)&=\sum_{r} \beta_r z^{-r-\frac32}, &
\gamma(z)&=\sum_{r} \gamma_r z^{-r+\frac12},
\end{alignat}
\end{subequations}
where $n\in{\mathbb Z}$ and $r\in{\mathbb Z}+\frac12$.
We also define a zero-mode $\alpha^\mu_0=\sqrt{2\ap}p^\mu$.
The (anti-) commutation relations of these operators are
\begin{subequations}
\begin{alignat}{2}
[x^\mu,p^\nu]&=i\eta^{\mu\nu},& &\\
 [\alpha_m^\mu,\alpha_n^\nu]&=m\eta^{\mu\nu}\delta_{n+m,0},&\qquad
 \{\psi^\mu_r,\psi^\nu_s\}&=\eta^{\mu\nu}\delta_{r+s,0},\\
 \{b_m,c_n\}&=\delta_{n+m,0},& [\gamma_r,\beta_s]&=\delta_{r+s,0},
\end{alignat}
\end{subequations}
with the space-time metric $\eta^{\mu\nu}={\rm diag}(-1,+1,\cdots,+1)$.
These oscillators satisfy the hermiticity relations
\begin{alignat}{3}
 (x^\mu)^\dag&=x^\mu\,,&\qquad (p^\mu)^\dag&=p^\mu\,,&\qquad\nonumber\\
 (\alpha^\mu_n)^\dag&=\alpha^\mu_{-n}\,,&\qquad (b_n)^\dag&=b_{-n}\,,&\qquad (c_n)^\dag&=c_{-n}\,,
\label{def:hermite_conjugate}\nn
 (\psi^\mu_r)^\dag&=\psi^\mu_{-r}\,,&\qquad
 (\beta_r)^\dag&=-\beta_{-r}\,,&\qquad (\gamma_r)^\dag&=\gamma_{-r}\,.
\end{alignat}

\section{General Picture and Its Inner Product}\label{app:two}

In this appendix, we consider the ground state of half-integer oscillator modes 
in general $l$-picture defined by
\begin{align}\label{def:l-pic_ground_state}
\begin{split}
 \beta_r\ket{l}&=0\, ,\qquad {\rm for}\quad r\ge -l-\frac12\, ,\\
 \gamma_r\ket{l}&=0\, ,\qquad {\rm for}\quad r\ge l+\frac32\, ,
\end{split}
\end{align}
with $l\in\mathbb{Z}$. Note that these \textit{ground} states 
are not the lowest-energy states except in natural picture
because a number of oscillators create negative energy excitations.
The ground state of the nonzero modes, introduced 
in (\ref{eq:0vac}) for $l=0$, is the direct product of this $\ket{l}$ and the ground state
$\ket{0}$ of integer oscillator modes defined by
\begin{equation}
 \alpha_n^\mu \ket{0}=b_n \ket{0}=c_n \ket{0} = 0 
\qquad {\rm for} \quad n > 0 \,, 
\end{equation}
as $\ket{0}_l=\ket{0}\otimes\ket{l}$.
Except in ($-1$)-picture, these ground states do not have a nonzero norm without a nontrivial metric
owing to the picture number anomaly.

A natural inner product is defined between the ground states in $l$- and $(-l-2)$-pictures:
\begin{equation}
 \braket{-l-2}{l}=1.\label{def:natural_inner_product}
\end{equation}
This inner product can be extended to the whole Fock space
because two Fock spaces are isomorphic and can be obtained by exchanging $\gamma_{-r}$
and $\beta_{-r}$ as seen from (\ref{def:l-pic_ground_state}). Two Fock spaces are dual
with respect to this natural inner product. In order to define the norm, however,
we cannot use this isomorphism because the BRST charge is not invariant under this exchange.

We can construct a map between the ground states in any two pictures using delta function
$\delta(\gamma_r)$ or $\delta(\beta_r)$. For example, the ground state in $(l\mp1)$-picture can be 
obtained from the one in $l$-picture as
\begin{subequations}\label{def:pictures}
\begin{align}
 \delta(\gamma_{l+\frac12})\ket{l}&=\ket{l-1},\\
 \delta(\beta_{-l-\frac32})\ket{l}&=\ket{l+1}.
\end{align}
\end{subequations}
One can easily show that this mapping is consistent with the condition (\ref{def:l-pic_ground_state}) 
using the identity 
\begin{subequations}\label{def:delta_relations}
\begin{alignat}{2}
\gamma_{l+\frac12}\delta(\gamma_{l+\frac12})&=0,&\qquad [\beta_{-l-\frac12},\delta(\gamma_{l+\frac12})]\ne0\,,\\
\beta_{-l-\frac32}\delta(\beta_{-l-\frac32})&=0,&\qquad [\gamma_{l+\frac32},\delta(\beta_{-l-\frac32})]\ne0\,.
\end{alignat}
\end{subequations}
The extension to the map between any two pictures is straightforward. 
Using these maps, the natural inner product (\ref{def:natural_inner_product}) can be rewritten
as an inner product between states in a picture with a nontrivial metric. 
\begin{subequations}\label{def:inner_product_with_metric}
\begin{alignat}{3}
 \bra{l}\delta(\gamma_{-l-\frac12})\cdots\delta(\gamma_{l+\frac12})\ket{l}&=1,&\qquad
{\rm for}&\qquad l>0\,,\\
 \bra{l}\delta(\beta_{l+\frac32})\cdots\delta(\beta_{-l-\frac32})\ket{l}&=1,&\qquad
{\rm for}&\qquad l<-1\,. 
\end{alignat}
\end{subequations}
One can also give a rigorous definition of these delta functions of oscillator modes
using the relations (\ref{def:pictures}) and the natural inner product 
(\ref{def:natural_inner_product}):
\begin{subequations}
\begin{align}
 \delta(\gamma_{l+\frac12})&=\ket{l-1}\bra{-l-2}\,,\\
 \delta(\beta_{-l-\frac32})&=\ket{l+1}\bra{-l-2}\,.
\end{align}
\end{subequations}

The delta function of a local field can also be defined using these delta functions of oscillator modes. 
Let us explain it by taking $\delta(\beta(z))$ as an example. This definition depends on the picture 
of the state on which the delta function acts. For example, if it acts on the state in natural picture,
it is defined by expanding around $\beta_{-\frac12}$ as
\begin{equation}
 \delta(\beta(z))=\sum_{n=0}^\infty\frac{1}{n!}\Big(\sum_{r\ne-\frac12}\beta_rz^{-r-\frac12}\Big)^n
z\delta(\beta_{-\frac12})^{(n)},
\end{equation}
with
\begin{align}
 \delta(\beta_{-\frac12})^{(n)}&=[\gamma_\frac12,[\gamma_\frac12,\cdots,[\gamma_\frac12,
\delta(\beta_{-\frac12})]]]\,,\nonumber\\
&=\sum_{m=0}^n\frac{(-1)^mn!}{(n-m)!m!}
(\gamma_\frac12)^{n-m}\ket{0}\bra{-1}(\gamma_\frac12)^m\,.
\end{align}

On the other hand, the theta function in (\ref{def:pco_X}) has no precise definition.
This is a realization of the fact that the picture-changing operator $X(z)$ is nontrivial,
while it is written in the BRST exact form.


\end{document}